\documentclass[11pt]{article}

\usepackage{graphicx}
\usepackage[utf8]{inputenc}
\usepackage[english]{babel}
\usepackage[pdftex]{hyperref}
\usepackage{times}
\usepackage{setspace}
\usepackage{xspace}
\usepackage{amsmath}
\usepackage{amsthm}
\usepackage{amsfonts}
\usepackage{amssymb}
\usepackage{verbatim}
\usepackage{scrextend}
\usepackage{fullpage}

\newtheorem{theorem}{Theorem}
\newtheorem{lemma}{Lemma}
\newtheorem{definition}{Definition}
\newtheorem{specification}{Specification}

\deffootnote[1em]{1em}{0em}{\textsuperscript{\thefootnotemark}\,}

\newcommand{\ie}{{\em i.e.,}\xspace}
\newcommand{\eg}{{\em e.g.,}\xspace}

\begin{document}

\title{The Next 700 Impossibility Results in Time-Varying Graphs}

\renewcommand*{\thefootnote}{\fnsymbol{footnote}}

\author{
Nicolas Braud-Santoni\footnotemark[1]
\and
Swan Dubois\footnotemark[2]
\and
Mohamed-Hamza Kaaouachi\footnotemark[2]
\and 
Franck Petit\footnotemark[2] $^,$\footnotemark[3]
}

\footnotetext[1]{IAIK, TU Graz, Graz, Austria}
\footnotetext[2]{Sorbonne Universit\'es, UPMC Universit\'e Paris 6, F-75005, Paris, France\\
CNRS, UMR 7606, LIP6, F-75005, Paris, France\\
Inria, \'Equipe-projet REGAL, F-75005, Paris, France}
\footnotetext[3]{Contact author, e-mail: franck.petit@lip6.fr}

\renewcommand*{\thefootnote}{\arabic{footnote}}
\setcounter{footnote}{0}

\date{}
\maketitle

\begin{abstract}
We address highly dynamic distributed systems modeled by time-varying graphs (TVGs). We interest in proof of impossibility results that often use informal arguments about convergence. First, we provide a distance among TVGs to define correctly the convergence of TVG sequences. Next, we provide a general framework that formally proves the convergence of the sequence of executions of any deterministic algorithm over TVGs of any convergent sequence of TVGs. Finally, we illustrate the relevance of the above result by proving that no deterministic algorithm exists to compute the underlying graph of any connected-over-time TVG, \ie any TVG of the weakest class of long-lived TVGs.
\end{abstract}

\section{Introduction}\label{sec:intro}

The availability of wireless communications has drastically increased in recent years and established new applications
that make various communicating agents and terminals (\eg robots, sensors, Unmanned Aerial Vehicles, ...) interact
together. A common feature of a vast majority of these networks is their {\em high
dynamic}, meaning that their topology keeps continuously changing over time.  
Classically, distributed systems are
modeled by a static undirected connected graph where vertices are processes and edges represent bidirectional
communication links.  Clearly, such modeling is not suitable for high dynamic networks.  

Numerous models taking in account topological changes over time have have been proposed since several decades, to
quote only a few, \cite{AKMUV12,AE84,CCF09,F03,F04,FGM07,SW09}. Some works aim at unifying most of the above
approaches. For instance, in~\cite{XFJ03}, the authors introduced the {\em evolving graphs}. They proposed modeling
the time as a sequence of discrete time instants and the system dynamic by a sequence of static graphs, one for each
time instant. More recently, another graph formalism, called {\em Time-Varying Graphs} (TVG), has been provided
in~\cite{CFQS12}. In contrast with evolving graphs, TVGs allow systems evolving within continuous time. Also
in~\cite{CFQS12} and in companion papers~\cite{CFMS10,CFMS12}, TVGs are gathered and ordered into classes depending
mainly on two main features: the quality of connectivity among the participating nodes and the
possibility/impossibility to perform tasks.

In distributed computing, impossibility results are difficult to prove formally. As an example, \emph{informal} arguments about convergence properties of sequences of objects (\eg graphs, executions, ...) are often used to prove such results (by building counter examples), that may weaken our confidence in their accuracy. In this paper, we propose a general framework that would help us for proving \emph{formally} impossibility results in TVGs. We first define a metric to compute a distance between any pair of TVGs based on the length of their longest common temporal prefix. Such distance allows to study the convergence of TVG sequences. Our main result consists in showing that, given an algorithm $\mathcal{A}$ designed for any TVG and a sequence of TVGs that converges toward a TVG $g$, then the sequence of executions of $\mathcal{A}$ on each TVG of the sequence also converges. Furthermore, the latter converges toward the execution of $\mathcal{A}$ over $g$. 

Next, we provide an example of use of this general result. We consider the weakest class of long-lived TVGs, in the
following referred to as \emph{connected-over-time} TVGs, \ie the class of TVGs where any node can contact any other node infinitely often. It can also be described as the family of TVGs such that the eventual underlying graph (\ie the subgraph encompassing all edges that are \emph{infinitely often} present) is connected. More precisely, we show that no deterministic algorithm exists to compute the eventual underlying graph of a connected-over-time TVG. This impossibility result intuitively comes from the fact that, with such an algorithm, no node is able to determine whether any of its adjacent edges (appearing/disappearing arbitrarily along the time) may disappear definitively or not.

Section~\ref{sec:TVG} presents the model. Our main result is presented in Section~\ref{sec:main}, followed by the impossibility result about underlying graph computation over the class of connected-over-time TGVs~ (Section~\ref{sec:COT}). Section~\ref{sec:conc} concludes this work. 

\section{Time-Varying Graph: Model}\label{sec:TVG}

This section aims to present formally the framework of our study of dynamic systems: time-varying graphs (TVGs). This model was introduced by \cite{CFQS12}. We present only definitions needed for the comprehension of our work.  We refer the reader to \cite{CFQS12} for more details 
and an interesting taxonomy of TVGs. 

\paragraph{Model.}\label{sub:model}

Let us first borrow the formalism introduced in~\cite{CFQS12} in order to describe the distributed systems prone to high dynamic. We consider {\em distributed systems} made of $n$ computing entities, henceforth indifferently referred to as {\em nodes}, {\em vertices}, or {\em processes}. A process has a local memory, a local sequential and deterministic algorithm, and input\slash output capabilities. All these entities are gathered in a set $V$.  Let $E$ be a set of edges (or relations) between pairwise entities, that describes interactions between processes, namely communication exchanges. The presence of an edge between two vertices $p$ and $q$ at a given time $t$ means that each vertex among $\{p,q\}$ is able to send a message to the other at $t$.

The interactions between processes are assumed to take place over a time span $\mathcal{T} \subseteq \mathbb{T}$ called the {\em lifetime} of the system. The temporal domain $\mathbb{T}$ is generally assumed to be either $\mathbb{N}$ (discrete-time systems) or $\mathbb{R}^+$ (continuous-time systems).

\begin{definition}[Time-varying graph \cite{CFQS12}]
\label{def:TVG}
A time-varying graph (TVG for short) $g$ is a tuple $(V,E,\mathcal{T},\rho,\zeta,\phi)$ where $V$ is a (static) set of vertices $\{v_1,\ldots,v_n\}$, $E$ a (static) set of edges between these vertices $E\subseteq V\times V$, $\rho:E\times\mathcal{T}\to\{0,1\}$ (called presence function) that indicates whether a given edge is available (\emph{i.e.} present) at a given time, $\zeta:E\times\mathcal{T}\rightarrow \mathbb{T}$ (called edge latency function) indicates the time it takes to cross a given edge if starting at a given date, and $\phi:V\times\mathcal{T}\rightarrow \mathbb{T}$ (called process latency function) indicates the time an internal action of a process takes at a given date.
\end{definition}

Given a TVG $g$, let $\mathcal{T}_g$ be the subset of $\mathcal{T}$ for which a topological event (appearance/disappearance of an edge) occurs in $g$. The evolution of $g$ during its lifetime $\mathcal{T}$ can be described as the sequence of graphs $\mathcal{S}_{g} = g_1, g_2,\ldots$, where $g_i=(V,E_i)$ corresponds to the static {\em snapshot} of $g$ at time $t_i \in \mathcal{T}_g$, \ie $e\in E_i$ if and only if $\forall t\in[t_i,t_{i+1}[,\rho(e,t) = 1$. Note that, by definition, $g_i \neq g_{i+1}$ for any $i$. 

We consider {\em asynchronous} distributed systems, \ie no pair of processes has access to any kind of shared device that could allow to synchronize their execution rate.  Furthermore, at any time, no process has access to the output of $\zeta$, \ie none of them can ({\em a priori}) predict a bound on the message delay. Note that the ability to send a message to another process at a given time does not mean that this message will be delivered. Indeed, the dynamicity of the communication graph implies that the edge between the two processes may disappear before the delivery of this message leading to the lost of messages in transit. 

The presences and absences of an edge are instantly detected by its two adjacent nodes. We assume that our system provides to each process a non-blocking communication primitive named \textbf{Send\_retry} that ensures the following property. When a process $p$ invokes \textbf{Send\_retry}$(m,q)$ (where $m$ is an arbitrary message and $q$ another process of $V$) at time $t$, this primitive delivers $m$ to $q$ in a finite time provided that there exists a time $t'\geq t$ such that the edge $\{p,q\}$ is present at time $t'$ during at least $\zeta(\{p,q\},t')$ units of time. In other words, the delivery of the message is ensured if there is, after the invocation of the primitive, an availability of the edge that is sufficient to overcome the communication delay of the edge at this time. Note that this primitive may never deliver a message (\emph{e.g.} if the considered edge never appears after invocation). Details of the implementation of this primitive are not considered here but it typically consists in resending $m$ at each apparition of the edge $\{p,q\}$ until its reception by $q$. This primitive allows us to abstract from topology changes and asynchronous communication and to write high-level algorithms.

\paragraph{Configurations and executions.} The \emph{state} of a process is defined by the values of its variables. Given a TVG $g$, a \emph{configuration} of $g$ is a vector of $n+2$ components $(g_i, M_i, p_1, p_2, \ldots, p_n)$ such that $g_i$ is a static snapshot of $g$ (\ie $g_i \in \mathcal{S}_{g}$), $M_i$ is the set of multi-sets of messages carried over $E_i$, and $p_1$ to $p_n$ represent the state of the $n$ processes in $V$. 

An {\em execution} of the distributed system modeled by $g$ is a sequence of configurations $e=\gamma_{0}, \ldots ,\gamma_{k},$ $\gamma_{k+1}, \ldots$, such that for each $k\geq 0$, during an execution step $(\gamma_k,\gamma_{k+1})$, one of the following event occurs: $(i)$ $g_{k} \neq g_{k+1}$, or $(ii)$ at least one process receives a message, sends a message, or executes some internal actions changing its state. The \emph{algorithm} executed by $g$ describes the set of all allowed internal actions of processes (in function of their current state or external events as message receptions or time-out expirations) during an execution of $g$. We assume that during any configuration step $(\gamma_{k},\gamma_{k+1})$ of an execution, if  $g_{k} \neq g_{k+1}$, then for each edge $e$ such that $e \in E_k$  and $e \notin E_{k+1}$ (\ie $e$ disappears during the step $(\gamma_k,\gamma_{k+1}$), none of the messages carried by $e$ belongs to $M_{k+1}$. Also, for each edge $e$ such that $e \in E_{k+1}$  and $e \notin E_{k}$ (\ie $e$ appears during the step $(\gamma_k,\gamma_{k+1})$), $e$ contains no message in configuration $\gamma_{k+1}$. 

\paragraph{Connected over time TVGs.} A key concept of time-varying graphs has been identified in~\cite{CFQS12}. The authors shows that the classical notion of path in static graphs in meaningless in TVGs. Indeed, some processes may communicate even if there is no (static) path between them at each time. To perform communication between two processes, the existence of a \emph{temporal path} between them is sufficient. They define such a temporal path as follows: a sequence of ordered pairs $\mathcal{J}=\{(e_1,t_1),(e_2,t_2),...,(e_k,t_k)\}$ such that $\{e_1,e_2,...,e_k\}$ is a path\footnote{A sequence of edges $\{v_1,v'_1\}, \{v_2,v'_2\}, \ldots, \{v_k,v'_k\}$ is a \emph{path} if  $\forall i \in\{1,k-1\}, v_{i+1} = v'_{i}$.} if for every $i \in [1,k]$, $\rho(e_i,t_i)=1$ and $t_{i+1} \geq t_i + \zeta(e_i,t_i)$. In other words, a temporal path from process $p$ to process $q$ is a sequence of adjacent edges from $p$ to $q$ such that availability and latency of edges allow the sending of a message from $p$ to $q$ using the \textbf{Send\_retry} primitive at each intermediate process (refer to \cite{CFQS12} for a formal definition). Note that a temporal path is a non symmetric relation between two processes.

Based on various assumptions made about temporal paths (\eg recurrence, periodicity, symmetry, and so on), the authors propose in~\cite{CFQS12} proposes a relevant hierarchy of TVG classes. In this paper, we choose to make minimal assumptions on the dynamicity of our system since we restrict ourselves on \emph{connected-over-time} TVGs defined as follows: 

\begin{definition}[Connected-over-time TVG \cite{CFQS12}]\label{def:COT}
A TVG $(V,E,\mathcal{T},\rho,\zeta,\phi)$ is connected-over-time if, for any time $t\in\mathcal{T}$ and for any pair of processes $p$ and $q$ of $V$, there exists a temporal path from $p$ to $q$ after time $t$. The class of connected-over-time TVGs is denoted by $\mathcal{COT}$\footnote{Authors of \cite{CFQS12} refer to this class as C5 in their hiearchy of TVG classes.}.
\end{definition}

Note that the lifetime of a connected-over-time TVG is necessarily infinite by definition. The class $\mathcal{COT}$ allows us to capture highly dynamic systems since we only require that any process will be always able to communicate with any other one without any supplementary assumption on this communication (such as delay, periodicity, or used route). In particular, note that a connected-over-time TVG may be disconnected at each time and that the presence of an edge at a given time does not preclude that this edge will appear again after this time.  Define an \emph{eventual missing edge} as en edge that appears only a finite number of time during the lifetime of the TVG. The main difficulty encountered in the design of distributed algorithms in $\mathcal{COT}$ is to deal with such eventual missing edges because no process is able to predict if a given adjacent edge is an eventual missing edge or not. Note that the time of the last presence of such an eventual missing edge cannot be even bounded.

\begin{definition}[(Eventual) Underlying Graph]\label{def:UG}
Given a TVG $g=(V,E,\mathcal{T},\rho,\zeta,\phi)$, the underlying graph of a $g$ is the (static) graph $U_g=(V,E)$. The eventual underlying graph of $g$ is the (static) subgraph $U^\omega_g=(V,E^\omega_g)$ with $E^\omega_g=E \setminus M_g$, where $M_g$ is the set of eventual missing edges of $g$. 
\end{definition}

Intuitively, the underlying graph (sometimes referred to as {\em footprint}) of a TVG $g$ gathers all edges that appear at least once during the lifetime of $g$, whereas the eventual underlying graph of $g$ gathers all edges that are infinitely often present during the lifetime of $g$. Note that, for any TVG of $\mathcal{COT}$, both underlying graph and eventual underlying graph are connected by definition. Let us define the \emph{neighborhood} $\mathcal{N}_p$ of a process $p$ is the set of processes with which $p$ shares an edge in the underlying graph.

\paragraph{Induced subclasses.}

In the following, we focus on specific subclasses of the class $\mathcal{COT}$ to establish our impossibility result. Informally, we focus on subclasses that gather all TVGs whose underlying graph belongs to a given set. The intuition behind this restriction is the following. In practice, some technical reasons may restrict or prevent the communication between some processes, that induces a given underlying graph for the TVG that models our system. In contrast, we cannot predict in general the availabilities and latencies of communication edges, that leads us to consider all TVGs sharing this underlying graph.

\begin{definition}[Induced subclass]\label{def:inducedsubclass}
Given a set of (static) graphs $\mathcal{F}$ and a class of TVGs $\mathcal{C}$, the subclass of $\mathcal{C}$ induced by $\mathcal{F}$ (denoted by $\mathcal{C}|_\mathcal{F}$) is the set of all TVGs of $\mathcal{C}$ whose underlying graph belongs to $\mathcal{F}$.
\end{definition}

The two following results follow directly from Definitions \ref{def:COT} and \ref{def:inducedsubclass}:

\begin{lemma}\label{lem:lem5}
In any induced subclass $\mathcal{C}|_\mathcal{F}$, if a TVG $g\in\mathcal{G}$ admits $f\in\mathcal{F}$ as underlying graph, then any other TVG of $\mathcal{C}$ that admits $f$ as underlying graph belongs to $\mathcal{C}|_\mathcal{F}$.
\end{lemma}

\begin{lemma}\label{lem:lem6}
No TVG of $\mathcal{COT}|_\mathcal{F}$ admits an eventual missing edge if and only if $\mathcal{F}$ contains only  trees. 
\end{lemma}

\section{Main Theorem}\label{sec:main}

In this section, we state our main result that provides a general framework for proving impossibility results in TVGs. First, we introduce in Section \ref{sub:spaces} some tools needed for the proof of our theorem. Namely, we prove that TVGs and executions sets may be seen as metric spaces with useful topological properties. Then, we prove our main result in Section \ref{sub:convergence}.

\subsection{TVG and Output Spaces}\label{sub:spaces}

\paragraph{TVG Space.} For a given time domain $\mathbb{T}$, a given static graph $(V,E)$ and a given latency function $\zeta$, let us consider the set $\mathcal{G}_{(V,E),\mathbb{T},\zeta}$ of all TVGs over $\mathbb{T}$ that admit $(V,E)$ as underlying graph and $\zeta$ as latency function. For the sake of clarity, we will omit the subscript $(V,E),\mathbb{T},\zeta$ and simply denote this set by $\mathcal{G}$. Remark that two distinct TVGs of $\mathcal{G}$ can be distinguished only by their presence function. For any TVG $g$ in $\mathcal{G}$, let us denote its presence function by $\rho_g$. We define now the following application $d_\mathcal{G}$ over $\mathcal{G}$:
\[\begin{array}{rrcl}
d_\mathcal{G} : & \mathcal{G}\times\mathcal{G} & \longrightarrow & [0,1] \\
    & (g,g')                       & \mapsto &
\left\{\begin{array}{cl}
0 & \text{if } g=g'\\
2^{-\lambda} & \text{else, with } \lambda = \text{Sup }\{ t\in\mathbb{T}|\forall t'\leq t,\forall e\in E,\rho_g(e,t') = \rho_{g'}(e,t')\}
\end{array}\right.
\end{array}\]

\begin{lemma}\label{lem:dgultrametric}
The application $d_\mathcal{G}$ is an ultrametric over $\mathcal{G}$, \emph{i.e.}
\begin{enumerate}
\item $\forall (g,g')\in\mathcal{G}^2, d_\mathcal{G}(g,g')=0 \Leftrightarrow g=g'$
\item $\forall (g,g')\in\mathcal{G}^2, d_\mathcal{G}(g,g')=d_\mathcal{G}(g',g)$
\item $\forall (g,g',g'')\in\mathcal{G}^3, d_\mathcal{G}(g,g'') \leq max(d_\mathcal{G}(g,g'), d_\mathcal{G}(g',g''))$
\end{enumerate}
\end{lemma} 

\begin{proof}
The two first properties follow directly from the definition of $d_\mathcal{G}$. 

To prove the third one, let $g$, $g'$, and $g''$ be three TVGs of $\mathcal{G}$. Assume that $d_\mathcal{G}(g,g')=2^{-\lambda'}$ and $d_\mathcal{G}(g',g'')=2^{-\lambda''}$ and let be $\lambda=min(\lambda',\lambda'')$. Then, by definition of $d_\mathcal{G}$, we have: $\forall t<\lambda',\forall e\in E,\rho_g(e,t)=\rho_{g'}(e,t)$ and $\forall t<\lambda'',\forall e\in E,\rho_{g'}(e,t)=\rho_{g''}(e,t)$. We can deduce that $\forall t<\lambda,\forall e\in E,\rho_g(e,t)=\rho_{g''}(e,t)$, that means that $d_\mathcal{G}(g,g'')\leq 2^{-\lambda}$.

On the other hand, we have: $max(d_\mathcal{G}(g,g'),d_\mathcal{G}(g',g''))=max(2^{-\lambda'},2^{-\lambda''})=2^{-\lambda}$. In conclusion, $d_\mathcal{G}(g,g'')\leq max(d_\mathcal{G}(g,g'),d_\mathcal{G}(g',g''))$, that ends the proof.
\end{proof}

In other words, we can consider $(\mathcal{G}, d_\mathcal{G})$ as a metric space (an ultrametric is a particular case
of metric) and associate to $(\mathcal{G}, d_\mathcal{G})$ the canonical topology, \emph{i.e.} the set of all open
balls induced by $d_\mathcal{G}$ over $\mathcal{G}$. This topological space have the following property that is useful
in the following.

\begin{lemma}\label{lem:gcomplete}
The metric space $(\mathcal{G}, d_\mathcal{G})$ is complete, \emph{i.e.} a sequence converges in $\mathcal{G}$ if and only if this sequence is Cauchy\footnote{Recall that a Cauchy sequence in a metric space $(S,d_S)$ is a sequence $(u_n)_{n\in\mathbb{N}}$ of $S$ whose oscillation converges to $0$. More formally, $\forall \epsilon\in\mathbb{R}^{*+},\exists k\in\mathbb{N}, \forall i\in\mathbb{N}, d_S(u_{k},u_{k+i})<\epsilon$}.
\end{lemma} 

\begin{proof}
Recall that, by definition of convergence, any convergent sequence is Cauchy. Hence, let $(g_n)_{n\in\mathbb{N}}$ be a Cauchy sequence in $\mathcal{G}$. We are going to prove that $(g_n)_{n\in\mathbb{N}}$ converges in $\mathcal{G}$. By definition of a Cauchy sequence, we have: $\forall \epsilon\in\mathbb{R}^{*+},\exists k\in\mathbb{N}, \forall i\in\mathbb{N}, d_\mathcal{G}(g_{k},g_{k+i})<\epsilon$. In particular, we have: $\forall\lambda\in\mathbb{T},\exists k\in\mathbb{N}, \forall i\in\mathbb{N}, d_\mathcal{G}(g_{k},g_{k+i})<2^{-\lambda}$. 

In the other hand, by definition of $d_\mathcal{G}$, we know that the existence of $\lambda_{(k,i)}\in\mathbb{T}$ such that $d_\mathcal{G}(g_k,g_{k+i})<2^{-\lambda_{(k,i)}}$ for $k\in\mathbb{N}$ and $i\in\mathbb{N}$ means that $\forall t<\lambda_{(k,i)},\forall e\in E,\rho_{g_k}(e,t)=\rho_{g_{k+i}}(e,t)$. Hence, we have:  $\forall\lambda\in\mathbb{T},\exists k\in\mathbb{N}, \forall i\in\mathbb{N},\forall t<\lambda,\forall e\in E,\rho_{g_k}(e,t)=\rho_{g_{k+i}}(e,t)$. Let $g_\omega\in\mathcal{G}$ be the TVG defined by $\forall\lambda\in\mathbb{T},\forall e\in E, \rho_{g_\omega}(e,\lambda)=\rho_{g_{k}}(e,\lambda)$.

 Let $\epsilon\in\mathbb{R}^{*+}$ and $\lambda$ be the smallest integer such that $2^{-\lambda}<\epsilon$. Then, we know that $\exists k\in\mathbb{N},\forall i\in\mathbb{N},\forall t<\lambda,\forall e\in E,\rho_{g_{k+i}}(e,t)=\rho_{g_k}(e,t)=\rho_{g_\omega}(e,t)$. We can deduce that: $\forall i\in\mathbb{N},d_\mathcal{G}(g_k,g_\omega)\leq2^{-\lambda}<\epsilon$. In other words, $(g_n)_{n\in\mathbb{N}}$ converges to $g_\omega\in\mathcal{G}$, that proves the completeness of $(\mathcal{G}, d_\mathcal{G})$.
\end{proof}

\paragraph{Output Space.} For a given algorithm $\mathcal{A}$ and a given TVG $g$, let us define the $(\mathcal{A},g)$-output as the function that associates to any time $t\in\mathbb{T}$ the state of $g$ at time $t$ when it executes $\mathcal{A}$. We say that $g$ is the supporting TVG of this output. Let us consider the set $\mathcal{O}_{\mathcal{A},\mathcal{G}}$ of all $(\mathcal{A},g)$-outputs over all TVGs $g$ of $\mathcal{G}$. For the sake of clarity, we will omit the subscript $\mathcal{A},\mathcal{G}$ and simply denote this set by $\mathcal{O}$. Remark that two distinct outputs of $\mathcal{O}$ can be distinguished only by their supporting TVG. For any output $o$ in $\mathcal{O}$, let us denote its supporting TVG by $g_o$. We define now the following application $d_\mathcal{O}$ over $\mathcal{O}$:

\[\begin{array}{rrcl}
d_\mathcal{O} : & \mathcal{O}\times\mathcal{O} & \longrightarrow & [0,1] \\
    & (o,o')                       & \mapsto &
\left\{\begin{array}{cl}
0 & \text{if } o=o'\\
2^{-\lambda} & \text{else, with } \lambda = \text{Sup }\{ t\in\mathbb{T}|\forall t'\leq t,o(t') = o'(t')\}
\end{array}\right.
\end{array}\]

Due to the similarity between the definition of $d_\mathcal{G}$ and $d_\mathcal{O}$, we can easily prove the following result:

\begin{lemma}\label{lem:doultrametric}
The application $d_\mathcal{O}$ is an ultrametric over $\mathcal{O}$.
\end{lemma} 

As previously, we can consider $(\mathcal{O}, d_\mathcal{O})$ as a metric space, associate to $(\mathcal{O}, d_\mathcal{O})$ the canonical topology and prove the following result:

\begin{lemma}\label{lem:ocomplete}
The metric space $(\mathcal{O}, d_\mathcal{O})$ is complete.
\end{lemma} 

\subsection{Convergence of Sequences of TVGs}\label{sub:convergence}

We are now ready to state our main result. Intuitively, this theorem ensures us that, if we take a sequence of TVGs with ever-growing common prefixes, then the sequence of corresponding outputs also converges. Moreover, we are able to describe the output to which it converges as the output that corresponds to the TVG that shares all commons prefixes of our TVGs sequence. This result is useful since it allows us to construct counter-example in the context of impossibility results. Indeed, it is sufficient to construct a TVG sequence (with ever-growing common prefixes) and to prove that their corresponding outputs violates the specification of the problem for ever-growing time to exhibit an execution that violates infinitely often the specification of the problem. More formally, we have:
 
\begin{theorem}\label{th:main}
For any deterministic algorithm $\mathcal{A}$, if a sequence $(g_n)_{n\in\mathbb{N}}$ of $\mathcal{G}$ converges to a given $g_\omega\in\mathcal{G}$, then the sequence $(o_n)_{n\in\mathbb{N}}$ of the $(\mathcal{A},g_n)$-outputs converges to $o_\omega\in\mathcal{O}$. Moreover, $o_\omega$ is the $(\mathcal{A},g_\omega)$-output.
\end{theorem}

\begin{proof}
Let $\mathcal{A}$ be a deterministic algorithm and $(g_n)_{n\in\mathbb{N}}$ be a sequence of $\mathcal{G}$ that converges to a given $g_\omega\in\mathcal{G}$. Then, let $(o_n)_{n\in\mathbb{N}}$ be the sequence of the $(\mathcal{A},g_n)$-outputs.

First, we are going to prove that $(o_n)_{n\in\mathbb{N}}$ converges in $\mathcal{O}$. As $\mathcal{O}$ is complete (see Lemma \ref{lem:ocomplete}), it is sufficient to prove that $(o_n)_{n\in\mathbb{N}}$ is a Cauchy sequence to obtain this result. Let $\epsilon\in\mathbb{R}^{*+}$. As $\mathcal{G}$ is also complete (see Lemma \ref{lem:gcomplete}), we know that $(g_n)_{n\in\mathbb{N}}$ is a Cauchy sequence and hence, we have by definition: $\exists k_\epsilon\in\mathbb{N}, \forall i\in\mathbb{N}, d_\mathcal{G}(g_{k_\epsilon},g_{k_\epsilon+i})<\epsilon$. 

In the other hand, by definition of $d_\mathcal{G}$, we know that the existence of $\lambda_{(k,i)}\in\mathbb{T}$ such that $d_\mathcal{G}(g_k,g_{k+i})=2^{-\lambda_{(k,i)}}$ for $k\in\mathbb{N}$ and $i\in\mathbb{N}$ means that $\forall t<\lambda_{(k,i)},\forall e\in E,\rho_{g_k}(e,t)=\rho_{g_{k+i}}(e,t)$. As $\mathcal{A}$ is deterministic, we can deduce that $\forall t<\lambda_{(k,i)},o_k(t)=o_{k+i}(t)$ (since $g_{o_n}=g_n$ for any $n\in\mathbb{N}$ by construction of $(o_n)_{n\in\mathbb{N}}$). Then, the definition of $d_\mathcal{O}$ implies that $d_\mathcal{O}(o_k,o_{k+i})\leq 2^{-\lambda_{(k,i)}}$. In other words, we can deduce that we have $\forall k\in\mathbb{N}, \forall i\in\mathbb{N},d_\mathcal{O}(o_k,o_{k+i})\leq d_\mathcal{G}(g_k,g_{k+i})$.

We can conclude that $\exists k_\epsilon\in\mathbb{N}, \forall i\in\mathbb{N}, d_\mathcal{O}(o_{k_\epsilon},o_{k_\epsilon+i})<\epsilon$. In conclusion, $(o_n)_{n\in\mathbb{N}}$ is a Cauchy sequence and then converges to $o\in\mathcal{O}$.

Let $o_\omega$ be the $(\mathcal{A},g_\omega)$-output. Then, we are going to prove that $o=o_\omega$. As $d_\mathcal{O}$ is an ultrametric (see Lemma \ref{lem:doultrametric}), we know that $0\leq d_\mathcal{O}(o,o_\omega) \leq \max(d_\mathcal{O}(o, o_n), d_\mathcal{O}(o_n, o_\omega))$ for any $n\in\mathbb{N}$. By that precedes, the sequence $(d_\mathcal{O}(o, o_n))_{n\in\mathbb{N}}$ converges to $0$. Due to the determinism of $\mathcal{A}$ and the completeness of $\mathcal{G}$ and $\mathcal{O}$, we can prove by a similar reasoning as above that $d_\mathcal{O}(o_n,o_\omega) \leq d_\mathcal{G}(g_n, g_\omega)$ for any $n\in\mathbb{N}$. The convergence of $(g_n)_{n\in\mathbb{N}}$ to $g_\omega$ implies that the sequence $(d_\mathcal{G}(g_n, g_\omega))_{n\in\mathbb{N}}$ converges to $0$. Then, the sequence $(d_\mathcal{O}(o_n, o_\omega))_{n\in\mathbb{N}}$ also converges to $0$ (since $d_\mathcal{O}(o_n,o_\omega)\geq 0$ for any $n\in\mathbb{N}$). Then, the sequence $(\max(d_\mathcal{O}(o, o_n), d_\mathcal{O}(o_n, o_\omega)))_{n\in\mathbb{N}}$ converges to $0$ that implies that $d_\mathcal{O}(o,o_\omega)=0$. As $d_\mathcal{O}$ is a metric, we can conclude that $o=o_\omega$, that ends the proof.
\end{proof}

\section{Impossibility of Eventual Underlying Graph Computation \label{sec:COT}}

In this section, we present an application of our main theorem by proving a natural impossibility result. Namely, we prove that it is impossible to compute the underlying graph of a connected-over-time TVG with a deterministic algorithm. Intuitively, this impossibility result comes from the fact that, with such an algorithm, no process is able to determine if, along its adjacent edges, there exists some eventual missing edges or not. The formal proof of this intuitive result is not as simple as one may think at first glance. 

Before presenting the impossibility result, we have to specify our problem. We say that a process $p$ outputs a value $v$ in a configuration $\gamma$ if one of its variable (called an output variable) has the value $v$ in $\gamma$.

\begin{specification}
An algorithm $\mathcal{A}$ satisfies the eventual underlying graph specification for a class of TVGs $\mathcal{C}$ if every execution $e=\gamma_0,\gamma_1,...$ on any TVG $g$ of $\mathcal{C}$ has a suffix $e_i=\gamma_i,\gamma_{i+1},...$ for a given $i\in\mathbb{N}$ such that each process outputs the eventual underlying graph of $g$ in any configuration of $e_i$.
\end{specification} 

We are now ready to prove the impossibility of eventual underlying graph in connected-over-time TVGs.

\begin{theorem}\label{th:COTCOT}
For any set of (static) graphs $\mathcal{F}$ that does not contain only trees, there exists no deterministic algorithm that satisfies the eventual underlying graph specification for $\mathcal{COT}|_\mathcal{F}$.
\end{theorem}

\begin{proof}
We define, for any TVG $g=(V,E,\mathcal{T},\rho,\zeta,\phi)$,  the TVG $g\oplus\{(e_1,\mathcal{T}_{e_1}),\ldots,(e_k,\mathcal{T}_{e_k})\}$ (with, for any $i\in\{0,\ldots,k\}$, $e_i\in E$ and $\mathcal{T}_{e_i}\subseteq\mathcal{T}$) as the TVG $(V,E,\mathcal{T},\rho',\zeta,\phi)$ with:
\[\rho'(e,t)=\begin{cases}
1 \text{ if } \exists i\in\{0,\ldots,k\}, e=e_i \text{ and } t\in\mathcal{T}_{e_i}\\
\rho(e,t) \text{ otherwise}
\end{cases}\]

By contradiction, assume that there exists a set of (static) graphs $\mathcal{F}$ that does not contain only trees such that there exists a deterministic algorithm $\mathcal{A}$ that satisfies the eventual underlying graph specification for $\mathcal{COT}|_\mathcal{F}$. In consequence, any process that executes $\mathcal{A}$ outputs a (static) graph at any time.

By Lemma \ref{lem:lem6}, we know that there exists $g\in\mathcal{COT}|_\mathcal{F}$ such that $g=(V,E,\mathcal{T},\rho,\zeta,\phi)$ admits at least one eventual missing edge $e$. We construct then a sequence $(g_n)_{n\in\mathbb{N}}$ of TVGs as follows. We set $g_0=g$ and we define inductively $g_i$ for any $i\in\mathbb{N}$ as follows---refer to Figure~\ref{fig:g_w}:

\begin{enumerate}
\item Consider the execution of $\mathcal{A}$ over $g_i$ and let $\eta_{i}\in\mathcal{T}\cup\{+\infty\}$ be the largest time where $e$ belongs to the graph outputted by some process of $V$ (remark that $\eta_i=+\infty$ if and only if $e$ belongs infinitely often to the outputted graph of at least one process);
\item Let $g'_i=g_i\oplus(e,\mathcal{T}\cap]\eta_i,+\infty[)$;
\item Consider the execution of $\mathcal{A}$ over $g'_i$ and let $\alpha_{i}\in\mathcal{T}\cup\{+\infty\}$ be the smallest time strictly greater than $\eta_i$ where $e$ belongs to the graph outputted by all process of $V$ (remark that $\alpha_i=+\infty$ if and only if $e$ never belongs simultaneously to the outputted graph of all processes $\eta_i=+\infty$);
\item Let $g_{i+1}=g_i\oplus(e,\mathcal{T}\cap]\eta_i,\alpha_i[)$.
\end{enumerate}

\begin{figure}
  \centering 
  \includegraphics[scale=0.41]{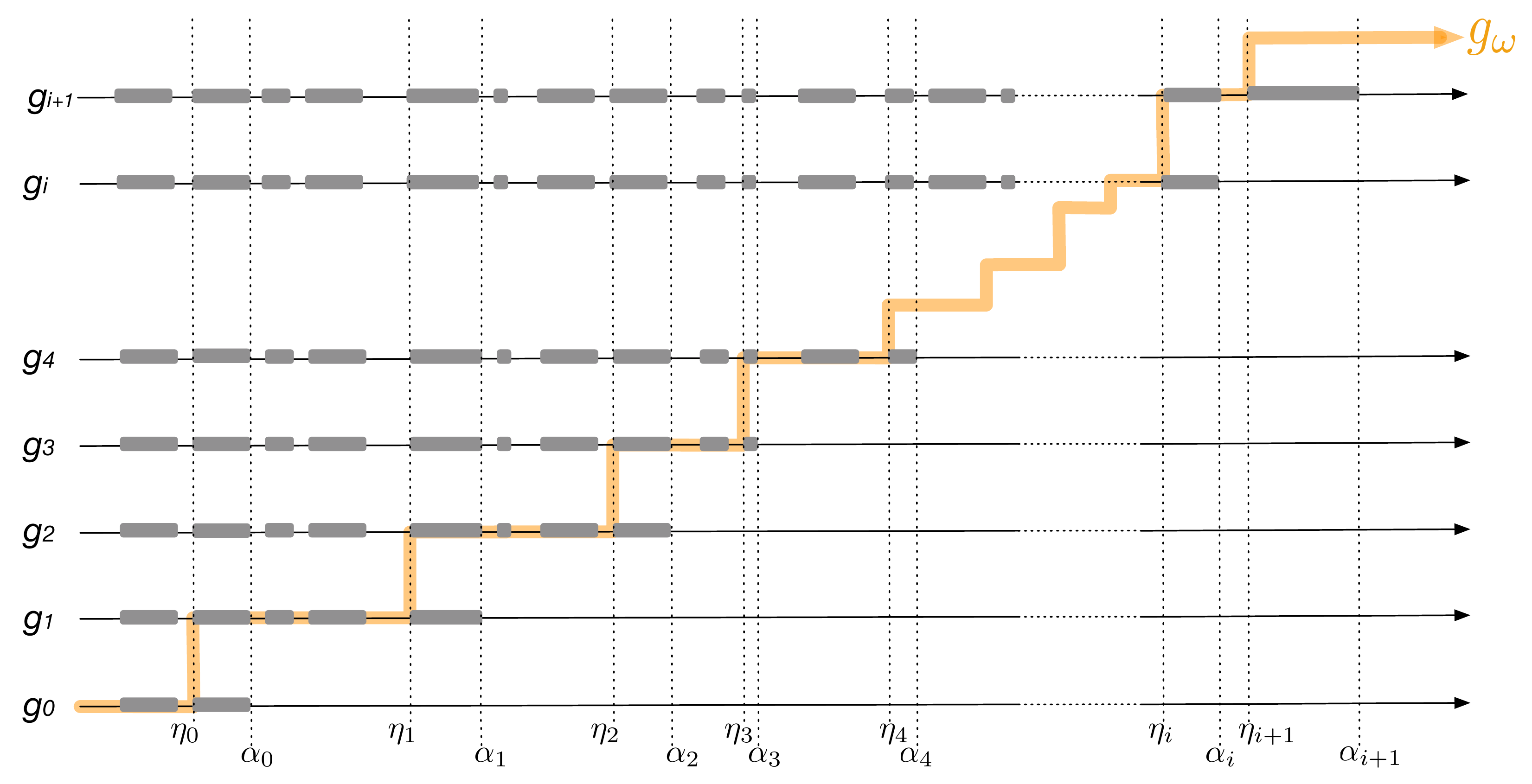}
  \caption{Construction of $(g_n)_{n\in\mathbb{N}}$ in the proof of Theorem~\ref{th:COTCOT}. Grey bold lines represent instants where $e$ belongs to the graph outputted by all process of $V$. \label{fig:g_w}}
\end{figure}

We can prove that, for any $i\in\mathbb{N}$, if $g_i$ belongs to $\mathcal{COT}|_\mathcal{F}$ and if $e$ is an eventual missing edge in $g_i$, then $\eta_i\neq+\infty$ and $\alpha_i\neq+\infty$. Indeed, assume that $e$ is an eventual missing edge in $g_i\in\mathcal{COT}|_\mathcal{F}$ for a given $i\in\mathbb{N}$. By definition, $e$ does not belong to $U^\omega_{g_i}$. As $\mathcal{A}$ satisfies the eventual underlying graph specification for $\mathcal{COT}|_\mathcal{F}$, we know that $e$ cannot belongs infinitely often to the outputted graph of a process in the execution of $\mathcal{A}$ over $g_i$, \ie $\eta_i\neq+\infty$. Then, as $e$ is not an eventual missing edge in $g'_i$ by construction, $e$ belongs to $U^\omega_{g'_i}$. By Lemma \ref{lem:lem5}, $g'_i$ belongs to $\mathcal{COT}|_\mathcal{F}$ since $g_i$ and $g'_i$ share the same underlying graph $U_g$. As $\mathcal{A}$ satisfies the eventual underlying graph specification for $\mathcal{COT}|_\mathcal{F}$, we know that $e$ belongs eventually to the outputted graph of all processes of $V$, \ie $\alpha_i\neq+\infty$.

We obtain that, for any $i\in\mathbb{N}$, if $g_i$ belongs to $\mathcal{COT}|_\mathcal{F}$ and if $e$ is an eventual missing edge in $g_i$, then $g_{i+1}$ belongs to $\mathcal{COT}|_\mathcal{F}$ and $e$ is an eventual missing edge in $g_{i+1}$. Indeed, $g_{i+1}$ belongs to $\mathcal{COT}|_\mathcal{F}$ by Lemma \ref{lem:lem5} (since $g_i$ and $g_{i+1}$ share the same underlying graph $U_g$). As we proved that $\eta_i\neq+\infty$ and $\alpha_i\neq+\infty$ when $e$ is an eventual missing edge in $g_i$, $g_{i+1}$ is obtained by adding $e$ during a finite amount of time to $g_i$, that implies that $e$ is an eventual missing edge in $g_{i+1}$.

Now, it is sufficient to note that $g$ belongs to $\mathcal{COT}|_\mathcal{F}$ by assumption and that $e$ is an eventual missing edge in $g_0=g$ by construction to obtain that $(g_n)_{n\in\mathbb{N}}$ is a sequence of $\mathcal{COT}|_\mathcal{F}$ such that $\eta_i\neq+\infty$ and $\alpha_i\neq+\infty$ for any $i\in\mathbb{N}$. Moreover, note that, for any $i\in\mathbb{N}$, $\eta_i<\alpha_i$ (by construction) and $\alpha_i<\eta_{i+1}$ (since $e$ belongs to the graph outputted by any process at time $\alpha_i$ in $g_{i+1}$ whereas $e$ does not belong to the graph outputted by any process at time $\eta_{i+1}$ in $g_{i+1}$). 

That allows us to define the following TVG: $g_\omega=g\oplus\{(e,\mathcal{T}\cap]\eta_i,\alpha_i[)|i\in\mathbb{N}\}$. Note that $U_{g_\omega}=U_g$ and then, by Lemma \ref{lem:lem5}, that $g_\omega$ belongs to $\mathcal{COT}|_\mathcal{F}$. Observe that, for any $k\in\mathbb{N}^*$, we have $d_\mathcal{G}(g_k,g_\omega)=2^{-\eta_k}$ by construction of $(g_n)_{n\in\mathbb{N}}$ and $g_\omega$. Thus, $(g_n)_{n\in\mathbb{N}}$ converges in $\mathcal{COT}|_\mathcal{F}$ to $g_\omega$.

We are now in measure to apply our main theorem (see Theorem \ref{th:main}) that states that the $(\mathcal{A},g_\omega)$-output is the limit of the sequence of the $(\mathcal{A},g_n)$-outputs. In other words, the $(\mathcal{A},g_\omega)$-output shares a prefix of length $\eta_i$ with the $(\mathcal{A},g_i)$-output for any $i\in\mathbb{N}$ (recall that the sequence of the $(\mathcal{A},g_n)$-outputs is Cauchy since it converges). That means that there exists infinitely many configurations in the execution of $\mathcal{A}$ on $g_\omega$ where $e$ belongs to the outputted graph of all process and infinitely many configurations in the execution of $\mathcal{A}$ on $g_\omega$ where $e$ does not belong to the outputted graph of any process, that contradicts the fact that $\mathcal{A}$ satisfies the eventual underlying graph specification for $\mathcal{COT}|_\mathcal{F}$ and ends the proof.
\end{proof}

\section{Conclusion}\label{sec:conc}

We gave a general framework for providing impossibility results in time-varying graphs. This framework is useful to legitimate informal arguments about convergence of sequences of objects in this context. We used the above result to prove that no deterministic algorithm exists to compute the underlying graph of any connected-over-time TVG. Our general framework is devoted to be used with a large number of problems in TVGs, \eg overlay construction.

\begin{small}
\bibliographystyle{plain}
\bibliography{IPL}

\begin{thebibliography}{10}

\bibitem{AKMUV12}
A.~Anagnostopoulos, R.~Kumar, M.~Mahdian, E.~Upfal, and F.~Vandin.
\newblock Algorithms on evolving graphs.
\newblock In {\em ITCS}, pages 149--160, 2012.

\bibitem{AE84}
B.~Awerbuch and S.~Even.
\newblock Efficient and reliable broadcast is achievable in an eventually
  connected network.
\newblock In {\em PODC}, pages 278--281, 1984.

\bibitem{CCF09}
A.~Casteigts, S.~Chaumette, and A.~Ferreira.
\newblock Characterizing topological assumptions of distributed algorithms in
  dynamic networks.
\newblock In {\em SIROCCO}, pages 126--140, 2009.

\bibitem{CFMS10}
A.~Casteigts, P.~Flocchini, B.~Mans, and N.~Santoro.
\newblock Deterministic computations in time-varying graphs: Broadcasting under
  unstructured mobility.
\newblock {\em Theoretical Computer Science}, pages 111--124, 2010.

\bibitem{CFMS12}
A.~Casteigts, P.~Flocchini, B.~Mans, and N.~Santoro.
\newblock Shortest, fastest, and foremost broadcast in dynamic networks.
\newblock Technical report, arXiv:1210.3277, 2012.

\bibitem{CFQS12}
A.~Casteigts, P.~Flocchini, W.~Quattrociocchi, and N.~Santoro.
\newblock Time-varying graphs and dynamic networks.
\newblock {\em International Journal of Parallel, Emergent and Distributed
  Systems}, 27(5):387--408, 2012.

\bibitem{F03}
K.~Fall.
\newblock A delay-tolerant network architecture for challenged internets.
\newblock In {\em SIGCOMM -- CATAPCC}, pages 27--34, 2003.

\bibitem{F04}
A.~Ferreira.
\newblock Building a reference combinatorial model for manets.
\newblock {\em Network}, 18(5):24--29, 2004.

\bibitem{FGM07}
A.~Ferreira, A.~Goldman, and J.~Monteiro.
\newblock On the evaluation of shortest journeys in dynamic networks.
\newblock In {\em NCA}, pages 3--10, 2007.

\bibitem{SW09}
J.~Schneider and R.~Wattenhofer.
\newblock Coloring unstructured wireless multi-hop networks.
\newblock In {\em PODC}, pages 210--219, 2009.

\bibitem{XFJ03}
B.~Xuan, A.~Ferreira, and A.~Jarry.
\newblock Computing shortest, fastest, and foremost journeys in dynamic
  networks.
\newblock {\em International Journal of Foundations of Computer Science},
  14(02):267--285, 2003.

\end{thebibliography}
\end{small}

\end{document}